\documentclass[12pt]{article}
\usepackage{amsmath}
\usepackage{amssymb}
\usepackage{amsthm}
\usepackage{graphicx}
\usepackage{lineno}
\usepackage{float}
\usepackage[all]{xy}
\usepackage{mathrsfs}

\usepackage{xcolor}
\usepackage{tikz}
\usetikzlibrary{tqft}

\newtheorem{theorem}{Theorem}
\newtheorem{lemma}{Lemma}
\newtheorem{corollary}{Corollary}

\newcommand{\be}{\begin{equation}}
\newcommand{\ee}{\end{equation}}

\newcommand{\bse}{\begin{subequations}}
\newcommand{\ese}{\end{subequations}}
\newcommand{\bea}{\begin{eqnarray}}
\newcommand{\eea}{\end{eqnarray}}
\newcommand{\ba}{\begin{array}}
\newcommand{\ea}{\end{array}}
\newcommand{\bc}{\begin{center}}
\newcommand{\ec}{\end{center}}

\usepackage{braket}

\begin{document}

\title{\bf Whether a quantum computation employs nonlocal resources is operationally undecidable}

\author{{Chris Fields$^1$, James F. Glazebrook$^2$, Antonino Marcian\`{o}$^3$ and Emanuele Zappala$^4$}\\ \\
{\it$^1$ Allen Discovery Center}\\
{\it Tufts University, Medford, MA 02155 USA}\\
{fieldsres@gmail.com} \\
{ORCID: 0000-0002-4812-0744} \\ \\
{\it$^2$ Department of Mathematics and Computer Science,} \\
{\it Eastern Illinois University, Charleston, IL 61920 USA} \\
{\it and} \\
{\it Adjunct Faculty, Department of Mathematics,}\\
{\it University of Illinois at Urbana-Champaign, Urbana, IL 61801 USA}\\
{jfglazebrook@eiu.edu}\\
{ORCID: 0000-0001-8335-221X}\\ \\
{\it$^3$ Center for Field Theory and Particle Physics \& Department of Physics} \\
{\it Fudan University, Shanghai, CHINA} \\
{\it and} \\
{\it Laboratori Nazionali di Frascati INFN, Frascati (Rome), ITALY} \\
{marciano@fudan.edu.cn} \\
{ORCID: 0000-0003-4719-110X} \\ \\
{\it$^4$ Department of Mathematics and Statistics,} \\
{\it Idaho State University, Pocatello, ID 83209, USA} \\
{emanuelezappala@isu.edu} \\
{ORCID: 0000-0002-9684-9441} \\ \\
}
\maketitle

\begin{abstract}
\noindent
Computational complexity characterizes the usage of spatial and temporal resources by computational processes. In the classical theory of computation, e.g. in the Turing Machine model, computational processes employ only local space and time resources, and their resource usage can be accurately measured by us as users. General relativity and quantum theory, however, introduce the possibility of computational processes that employ nonlocal spatial or temporal resources. While the space and time complexity of classical computing can be given a clear operational meaning, this is no longer the case in any setting involving nonlocal resources. In such settings, theoretical analyses of resource usage cease to be reliable indicators of practical computational capability.
We prove that the verifier (C) in a multiple interactive provers with shared entanglement (MIP*) protocol cannot operationally demonstrate that the ``multiple’’ provers are independent, i.e. cannot operationally distinguish a MIP* machine from a monolithic quantum computer. Thus C cannot operationally distinguish a MIP* machine from a quantum TM, and hence cannot operationally demonstrate the solution to arbitrary problems in RE. 
Any claim that a MIP* machine has solved a TM-undecidable problem is, therefore, circular, as the problem of deciding whether a physical system is a MIP* machine is itself TM-undecidable. 
Consequently,  despite the space and time complexity of classical computing having a clear operational meaning, this is no longer the case in any setting involving nonlocal resources. In such settings, theoretical analyses of resource usage cease to be reliable indicators of practical computational capability. This has practical consequences when assessing newly proposed computational frameworks based on quantum theories. 

\end{abstract}


\section{Introduction}

Computational complexity characterizes the usage of spatial and temporal resources by computational processes.  As users of such processes, we are interested in their resource requirements as measured by us. For example, we want to know whether a computation will halt in polynomial time as measured by our clocks.  In the classical theory of computation, e.g. in the Turing Machine (TM) model, computational processes employ only local space and time resources, and their resource usage can be accurately measured by us as users.  General relativity and quantum theory, however, introduce the possibility of computational processes that employ nonlocal spatial or temporal resources.  One notable example is the use of closed timelike curves (CTCs), which enable even otherwise-classical computers to employ arbitrary temporal resources as measured in their reference frames, and hence to solve problems that are exponential in time (class NEXP) for TMs \cite{deutsch:91, brun:03, aaronson:09}.  A second example is the ability of multiple, otherwise-independent, interactive provers (MIP) that share entanglement as a resource (MIP*) to solve, with probability approaching unity, TM-undecidable problems such as the Halting Problem (class RE). This is the celebrated result stating that MIP* = RE \cite{ji:20}.

We show in what follows that whether physically-implemented computational processes, i.e. physical computers, employ such nonlocal resources is operationally undecidable.  In the equivalent game-theoretic language, we show that whether players in a nonlocal game employ nonlocal strategies is undecidable by the referee of the game.  We demonstrate these results in the generic context of a {\em local operations, classical communication}  (LOCC) protocol \cite{chit:14}, in which quantum systems, interpretable as ``agents'' or ``processes'' or ``players'' Alice ($A$) and Bob ($B$) communicate via both quantum and classical channels traversing an environment ($E$), and in which the classical communication channel is via a third quantum system, interpretable as a ``user'' or ``verifier'' or ``referee'' Charlie ($C$), who is able to turn on, or off, an interaction that decoheres the quantum channel between $A$ and $B$.  Canonical Bell/EPR experiments in which $C$ both controls the source of entangled pairs observed by $A$ and $B$, and tests the observations recorded by $A$ and $B$ for violations of the Clauser-Horne-Shimony-Holt (CHSH) inequality \cite{chsh:69} have this form \cite{fgm:24a}.

Following a brief review of the relevant background in \S \ref{background}, we begin by showing in \S \ref{main-result} that $C$ cannot operationally demonstrate, using just data received from $A$ and $B$, that the joint state $|AB \rangle$ is separable.  From this it immediately follows that the verifier ($C$) in a multiple interactive provers with shared entanglement (MIP*) protocol cannot operationally demonstrate that the ``multiple'' provers are independent, i.e. cannot operationally distinguish MIP* machines from monolithic quantum computers.  As the latter are known to be TM-equivalent \cite{deutsch:85}, this shows that $C$ cannot operationally distinguish a MIP* machine from a quantum TM, and hence cannot operationally demonstrate the solution to arbitrary problems in RE.  Expressed in the language of {\em constraint satisfaction problems} (CSPs) \cite{culf:24}, $C$ cannot operationally demonstrate independence between constraints, and hence cannot operationally identify partial solutions.  We then employ the limit as $C \rightarrow E$ to show, in \S \ref{ctcs}, that a channel from $A$ to $B$ that is classical, and therefore causal, in the spacetime coordinates employed by $C$ may be a CTC in the coordinates employed by the joint system $AB$.  Hence $C$ cannot operationally determine whether computations implemented by $AB$ employ CTCs as a resource.  We conclude in \S \ref{concl} that while the space and time complexity of classical computing can be given a clear operational meaning, this is no longer the case in any setting involving nonlocal resources.  In such settings, therefore, theoretical analyses of resource usage cease to be reliable indicators of practical computational capability. 

\section{Background}\label{background}

\subsection{Nonlocal games}

We commence by recalling what is meant by a {\em  nonlocal game}, a concept commanding a special status in quantum information theory. A nonlocal game, in its basic form, unfolds via the interaction of three parties: two noncommuting {\em players} or {\em provers} $A$ and $B$ and a {\em verifier} or {\em referee} $C$. The players $A$ and $B$ are allowed to communicate classically before the start of play, but not after; they are also allowed to share an arbitrary, bipartite state. A verifier $C$ samples a pair of questions from some distribution, and then sends one of them to each of $A$ and $B$ separately. Each of $A$ and $B$ answers {\em classically} to the verifier. They win the game if the questions and answers satisfy a given predicate. Each of $A$ and $B$ knows the distribution of the questions and the predicate. The {\em quantum value} is the supremum of the probability that the players win the game (for generalizations to fully nonlocal quantum games with allowable noise and further details, see e.g. \cite{ji:20,vidick:22,qin:23}).

The above description can be extended from two provers to multiple provers. In a {\em multiple interactive prover} (MIP) game, first introduced in \cite{ben:88}, we have multiple provers who are able to communicate with each other prior to a problem being posed but not after, that try to convince a polynomial time verifier that a string $x$ belongs to some language $\mathcal {L}$. The class MIP($p$, $k$) indicates $p$ players with $k$ rounds.  It has been shown that considering two provers, i.e. $p = 2$ is always seen as sufficient; hence all such games can be represented in
MIP($2$, $k$) (often with $k=1$) \cite{ben:88,feige:92}. If shared entanglement is permitted, then we arrive at the class MIP* introduced in \cite{cleve:04}. As pointed out in \cite{culf:24} for the case of CSPs, entanglement permits provers to execute correlations that cannot be sampled by classical provers, i.e to violate the CHSH inequality. This improved ability on the part of the provers encourages the verifier to set harder tasks. A one-round MIP or MIP* is equivalent to a family of nonlocal games\footnote{In \cite{kalai:23} it is shown how a large class of multiprover, nonlocal games, can be recompiled/reduced to a single-prover interactive game. In the presence of a TM, without loss of generality, the game in question can be the one generated by the TM. A large class of such games are known to be undecidable (as discussed in \cite{fg:24} and references therein). }.

The ground-breaking result of Ji et al. \cite{ji:20} is that MIP*=RE, the latter being the {\em class of recursively enumerable languages}, i.e. the class of languages $\mathcal{L}$ equivalent to the Halting problem \cite{hopcroft:79}. In terms of quantum values, as noted in \cite {qin:23}, a consequence of the result of \cite{ji:20} is that approximating the quantum value of a fully nonlocal game is undecidable.  Crucially, the operational configuration that is employed in \cite{ji:20} to define a MIP* machine is a LOCC protocol: the two independent, and therefore separable, provers ($A$ and $B$) communicate classically via a TM (or user) verifier $C$ that poses problems and checks answers while sharing an entangled pair as a quantum communication channel $Q$.  Effectively, it is sufficient to prove the main result for MIP*(2,1), i.e. for two provers in one round.

\subsection{LOCC protocols}

The canonical LOCC protocol is a Bell/EPR experiment, where $A$ and $B$ must agree, via classical communication, to employ specified detectors in specified ways, and must later exchange their accumulated data (or transfer to the 3rd party $C$) in the form of classical records.  We showed in \cite{fgm:22a} that sequentially-repeated state preparations and/or measurements that employ mutually commuting QRFs (for instance, the sequentially repeated preparations and/or measurements of position and spin during a Bell/EPR experiment), are representable, without loss of generality, by topological quantum field theories (TQFTs) \cite{atiyah:88}. We then showed in \cite{fgm:24a} that any LOCC protocol can be represented by Diagram \eqref{locc-diag}, in which $A$ and $B$ are mutually separable and are separated from their joint environment $E$ by a holographic screen $\mathscr{B}$, implement read/write quantum reference frames (QRFs) $Q_A$ and $Q_B$, respectively, and communicate via classical and quantum channels implemented by $E$.

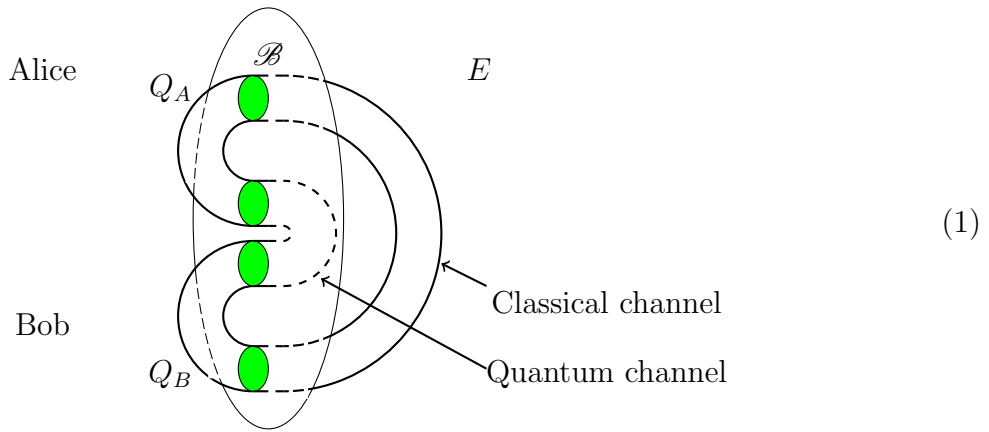
\begin{equation} \label{locc-diag}
\begin{gathered}
\begin{tikzpicture}[every tqft/.append style={transform shape}]
\draw[rotate=90] (0,0) ellipse (2.8cm and 1 cm);
\node[above] at (0,1.9) {$\mathscr{B}$};
\draw [thick] (-0.2,1.9) arc [radius=1, start angle=90, end angle= 270];
\draw [thick] (-0.2,1.3) arc [radius=0.4, start angle=90, end angle= 270];
\draw[rotate=90,fill=green,fill opacity=1] (1.6,0.2) ellipse (0.3 cm and 0.2 cm);
\draw[rotate=90,fill=green,fill opacity=1] (0.2,0.2) ellipse (0.3 cm and 0.2 cm);
\draw [thick] (-0.2,-0.3) arc [radius=1, start angle=90, end angle= 270];
\draw [thick] (-0.2,-0.9) arc [radius=0.4, start angle=90, end angle= 270];
\draw[rotate=90,fill=green,fill opacity=1] (-0.6,0.2) ellipse (0.3 cm and 0.2 cm);
\draw[rotate=90,fill=green,fill opacity=1] (-2.0,0.2) ellipse (0.3 cm and 0.2 cm);
\draw [rotate=180, thick, dashed] (-0.2,0.9) arc [radius=0.7, start angle=90, end angle= 270];
\draw [rotate=180, thick, dashed] (-0.2,0.3) arc [radius=0.1, start angle=90, end angle= 270];
\draw [thick] (-0.2,0.5) -- (0,0.5);
\draw [thick] (-0.2,-0.1) -- (0,-0.1);
\draw [thick] (-0.2,-0.9) -- (0,-0.9);
\draw [thick] (-0.2,-0.3) -- (0,-0.3);
\draw [thick, dashed] (0,0.5) -- (0.2,0.5);
\draw [thick, dashed] (0,-0.1) -- (0.2,-0.1);
\draw [thick, dashed] (0,-0.9) -- (0.2,-0.9);
\draw [thick, dashed] (0,-0.3) -- (0.2,-0.3);
\node[above] at (-3,1.7) {Alice};
\node[above] at (-3,-1.7) {Bob};
\node[above] at (2.8,1.7) {$E$};
\draw [ultra thick, white] (-0.9,1.5) -- (-0.7,1.5);
\draw [ultra thick, white] (-1,1.3) -- (-0.8,1.3);
\draw [ultra thick, white] (-1,1.1) -- (-0.8,1.1);
\draw [ultra thick, white] (-1,0.9) -- (-0.8,0.9);
\draw [ultra thick, white] (-1.1,0.7) -- (-0.8,0.7);
\draw [ultra thick, white] (-1.1,0.5) -- (-0.8,0.5);
\draw [ultra thick, white] (-1,-0.9) -- (-0.8,-0.9);
\draw [ultra thick, white] (-1,-1.1) -- (-0.8,-1.1);
\draw [ultra thick, white] (-1,-1.3) -- (-0.8,-1.3);
\draw [ultra thick, white] (-0.9,-1.5) -- (-0.7,-1.5);
\draw [ultra thick, white] (-0.9,-1.7) -- (-0.7,-1.7);
\draw [ultra thick, white] (-0.8,-1.9) -- (-0.6,-1.9);
\draw [ultra thick, white] (-0.8,-2.1) -- (-0.6,-2.1);
\node[above] at (-1.3,1.4) {$Q_A$};
\node[above] at (-1.3,-2.4) {$Q_B$};
\draw [rotate=180, thick] (-0.2,2.3) arc [radius=2.1, start angle=90, end angle= 270];
\draw [rotate=180, thick] (-0.2,1.7) arc [radius=1.5, start angle=90, end angle= 270];
\draw [thick] (-0.2,1.9) -- (0,1.9);
\draw [thick] (-0.2,1.3) -- (0,1.3);
\draw [thick, dashed] (0.2,1.9) -- (0,1.9);
\draw [thick, dashed] (0.2,1.3) -- (0,1.3);
\draw [thick] (-0.2,-1.7) -- (0,-1.7);
\draw [thick] (-0.2,-2.3) -- (0,-2.3);
\draw [thick, dashed] (0.2,-1.7) -- (0,-1.7);
\draw [thick, dashed] (0.2,-2.3) -- (0,-2.3);
\draw [ultra thick, white] (0.3,2) -- (0.3,1.2);
\draw [ultra thick, white] (0.5,2) -- (0.5,1.2);
\draw [ultra thick, white] (0.7,1.9) -- (0.7,1.1);
\draw [ultra thick, white] (0.3,-2.4) -- (0.3,-1.5);
\draw [ultra thick, white] (0.5,-2.4) -- (0.5,-1.5);
\draw [ultra thick, white] (0.7,-1.8) -- (0.7,-1.5);
\node[above] at (4.5,-2.4) {Quantum channel};
\draw [thick, ->] (2.9,-2) -- (0.7,-0.8);
\node[above] at (4.5,-1.4) {Classical channel};
\draw [thick, ->] (2.9,-0.9) -- (2.3,-0.6);
\end{tikzpicture}
\end{gathered}
\end{equation}
Note both that $A$ and $B$ being mutually separable is required for the assumption of classical communication via a causal channel in $E$, and that this assumption renders $Q_A$ and $Q_B$ noncommutative and hence subject to quantum contextuality \cite{fg:21,fg:23}.

Two defining characteristics of LOCC protocols are worth emphasizing \cite{fgmz:24}:
\begin{itemize}
    \item[(1)] $A$ and $B$ both perform only {\em local} operations.  They must, therefore, each employ spatial quantum reference frames (QRFs \cite{bartlett:07}), which we will denote $X_A$ and $X_B$, respectively, with respect to which they specify the position of the quantum degrees of freedom that they manipulate, e.g. the positions of the detectors in a Bell/EPR experiment.  These spatial QRFs must commute with the QRFs $Q_A$ and $Q_B$ that they, respectively, employ to manipulate the quantum channel, i.e. $[X_A, Q_A] = [X_B, Q_B] =_{def} 0$.

    \item[(2)]
$A$ and $B$ must both comprise sufficient degrees of freedom for them both to implement their respective QRFs and to communicate classically.  This is, effectively, a large $N$-limit that assures their separability as physical systems.
\end{itemize}

We proved in \cite[Thm. 1]{fgmz:24} that in the operational setting of a two-agent LOCC protocol \cite{chit:14}, two potential provers, $A$ and $B$, cannot operationally distinguish monogamous entanglement from a topological identification of points in their respective local spacetimes, one local to $A$, and the other local to $B$.  Specifically:

\begin{theorem} \label{main-th}
{\rm \cite[Thm. 1]{fgmz:24}}  In any LOCC protocol in which all systems are finite, and in which the boundary $\mathscr{B}$ between the communicating agents $A$ and $B$ and their joint environment $E$ is a holographic screen, as the entanglement made available to $A$ and $B$ by the quantum channel approaches pairwise monogamy, and hence the decoherence in the quantum channel detectable by $A$ or $B$ decreases to zero, the number of environmental degrees of freedom of $E$ required to implement the quantum channel becomes operationally indistinguishable, by $A$ or $B$, from zero in the limit of monogamous entanglement.
\end{theorem}

The proof is straightforward, and can be sketched as follows.  Let $q_A$ and $q_B$ be distinct (collections of) qubits accessible only to $A$ and $B$, respectively, and suppose $|q_Aq_B\rangle \neq |q_A\rangle |q_B\rangle$, i.e. there is a quantum channel $Q$ shared by $A$ and $B$.  If this channel is embedded in $E$ as shown in Diagram \eqref{locc-diag}, then we can consider the interaction $H_{Q \bar{Q}}$, where $\bar{Q}$ is the complement of $Q$ in $E$, i.e. $Q \bar{Q} = E$.  Monogamous entanglement of $q_A$ and $q_B$ requires that $Q$ be decoherence-free, i.e. that $H_{Q \bar{Q}} \rightarrow 0$.  This can be achieved topologically by folding the boundary $\mathscr{B}$ in a way that decreases the degrees of freedom of $E$ used to implement $Q$ to zero, in which case $Q$ is simply the joint state $|q_Aq_B\rangle$; see \cite{fgmz:24} for details.

Theorem \ref{main-th} has two significant corollaries as noted in \cite{fgmz:24}:

\begin{corollary} \label{qecc-cor}
The codespace dimension of a perfect QECC is operationally indistinguishable from the code dimension.
\end{corollary}

\begin{corollary} \label{er-epr-cor}
In any LOCC protocol in which all systems are finite, and in which the boundary $\mathscr{B}$ between the communicating agents $A$ and $B$ and their joint environment $E$ is a holographic screen, a quantum channel implementing a shared, monogamously-entangled pair of qubits (``EPR'') is operationally indistinguishable from a topological identification of the locally-measured locations $x_A$ and $x_B$ of the qubits accessed by $A$ and $B$ respectively (``ER'').
\end{corollary}

Hence the acclaimed hypothesis ER = EPR \cite{maldecena:13} can be recovered as an operational theorem, free of any embedding geometry, with the consequence that the local topology of spacetime is observer-relative, and providing a straightforward demonstration of the non-traversability of ER bridges.

\section{MIP* machines are not operationally identifiable}\label{main-result}

To apply Theorem \ref{main-th} to the operational context of a verifier $C$ interacting with a MIP* machine, we add $C$ to Diagram \eqref{locc-diag} as follows:

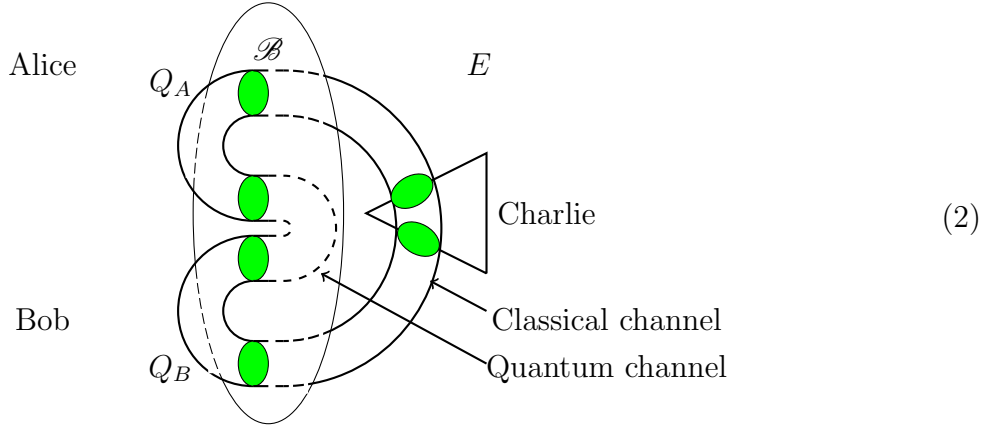
\begin{equation} \label{locc-C-diag}
\begin{gathered}
\begin{tikzpicture}[every tqft/.append style={transform shape}]
\draw[rotate=90] (0,0) ellipse (2.8cm and 1 cm);
\node[above] at (0,1.9) {$\mathscr{B}$};
\draw [thick] (-0.2,1.9) arc [radius=1, start angle=90, end angle= 270];
\draw [thick] (-0.2,1.3) arc [radius=0.4, start angle=90, end angle= 270];
\draw[rotate=90,fill=green,fill opacity=1] (1.6,0.2) ellipse (0.3 cm and 0.2 cm);
\draw[rotate=90,fill=green,fill opacity=1] (0.2,0.2) ellipse (0.3 cm and 0.2 cm);
\draw [thick] (-0.2,-0.3) arc [radius=1, start angle=90, end angle= 270];
\draw [thick] (-0.2,-0.9) arc [radius=0.4, start angle=90, end angle= 270];
\draw[rotate=90,fill=green,fill opacity=1] (-0.6,0.2) ellipse (0.3 cm and 0.2 cm);
\draw[rotate=90,fill=green,fill opacity=1] (-2.0,0.2) ellipse (0.3 cm and 0.2 cm);
\draw [rotate=180, thick, dashed] (-0.2,0.9) arc [radius=0.7, start angle=90, end angle= 270];
\draw [rotate=180, thick, dashed] (-0.2,0.3) arc [radius=0.1, start angle=90, end angle= 270];
\draw [thick] (-0.2,0.5) -- (0,0.5);
\draw [thick] (-0.2,-0.1) -- (0,-0.1);
\draw [thick] (-0.2,-0.9) -- (0,-0.9);
\draw [thick] (-0.2,-0.3) -- (0,-0.3);
\draw [thick, dashed] (0,0.5) -- (0.2,0.5);
\draw [thick, dashed] (0,-0.1) -- (0.2,-0.1);
\draw [thick, dashed] (0,-0.9) -- (0.2,-0.9);
\draw [thick, dashed] (0,-0.3) -- (0.2,-0.3);
\node[above] at (-3,1.7) {Alice};
\node[above] at (-3,-1.7) {Bob};
\node[above] at (2.8,1.7) {$E$};
\draw [ultra thick, white] (-0.9,1.5) -- (-0.7,1.5);
\draw [ultra thick, white] (-1,1.3) -- (-0.8,1.3);
\draw [ultra thick, white] (-1,1.1) -- (-0.8,1.1);
\draw [ultra thick, white] (-1,0.9) -- (-0.8,0.9);
\draw [ultra thick, white] (-1.1,0.7) -- (-0.8,0.7);
\draw [ultra thick, white] (-1.1,0.5) -- (-0.8,0.5);
\draw [ultra thick, white] (-1,-0.9) -- (-0.8,-0.9);
\draw [ultra thick, white] (-1,-1.1) -- (-0.8,-1.1);
\draw [ultra thick, white] (-1,-1.3) -- (-0.8,-1.3);
\draw [ultra thick, white] (-0.9,-1.5) -- (-0.7,-1.5);
\draw [ultra thick, white] (-0.9,-1.7) -- (-0.7,-1.7);
\draw [ultra thick, white] (-0.8,-1.9) -- (-0.6,-1.9);
\draw [ultra thick, white] (-0.8,-2.1) -- (-0.6,-2.1);
\node[above] at (-1.3,1.4) {$Q_A$};
\node[above] at (-1.3,-2.4) {$Q_B$};
\draw [rotate=180, thick] (-0.2,2.3) arc [radius=2.1, start angle=90, end angle= 270];
\draw [rotate=180, thick] (-0.2,1.7) arc [radius=1.5, start angle=90, end angle= 270];
\draw [thick] (-0.2,1.9) -- (0,1.9);
\draw [thick] (-0.2,1.3) -- (0,1.3);
\draw [thick, dashed] (0.2,1.9) -- (0,1.9);
\draw [thick, dashed] (0.2,1.3) -- (0,1.3);
\draw [thick] (-0.2,-1.7) -- (0,-1.7);
\draw [thick] (-0.2,-2.3) -- (0,-2.3);
\draw [thick, dashed] (0.2,-1.7) -- (0,-1.7);
\draw [thick, dashed] (0.2,-2.3) -- (0,-2.3);
\draw [ultra thick, white] (0.3,2) -- (0.3,1.2);
\draw [ultra thick, white] (0.5,2) -- (0.5,1.2);
\draw [ultra thick, white] (0.7,1.9) -- (0.7,1.1);
\draw [ultra thick, white] (0.3,-2.4) -- (0.3,-1.5);
\draw [ultra thick, white] (0.5,-2.4) -- (0.5,-1.5);
\draw [ultra thick, white] (0.7,-1.8) -- (0.7,-1.5);
\node[above] at (4.5,-2.4) {Quantum channel};
\draw [thick, ->] (2.9,-2) -- (0.7,-0.8);
\node[above] at (4.5,-1.7) {Classical channel};
\draw [thick, ->] (2.9,-1.3) -- (2.15,-0.9);
\draw [thick] (2.9, -0.8) -- (2.9, 0.8) -- (1.3, 0.0) -- (2.9, -0.8);
\draw[rotate=30,fill=green,fill opacity=1] (1.8, -0.7) ellipse (0.3 cm and 0.2 cm);
\draw[rotate=-30,fill=green,fill opacity=1] (1.9, 0.7) ellipse (0.3 cm and 0.2 cm);
\node at (3.7, 0.0) {Charlie};
\end{tikzpicture}
\end{gathered}
\end{equation}

Here $C$ interacts with $A$ and $B$ separately, and only via a classical channel, as required by the definition of MIP*.

We assume for convenience that $A$ and $B$ interact, respectively, with $q_A$ and $q_B$ in a computational basis in which single-qubit measurements have eigenvalues in $\{+1, -1 \}$; no generality is lost in also assuming that $q_A$ and $q_B$ are each single qubits.  The data items $A_i$ and $B_i$ reported by $A$ and $B$, respectively, using the classical channel always, therefore, have values in $\{+1, -1 \}$.  We also assume that $C$ has sufficient degrees of freedom, and in particular, access to sufficient classical memory, to collect sufficient classical data from both $A$ and $B$ to compute the CHSH expectation value with negligible uncertainty.  The CHSH expectation value is:

\begin{equation*} \label{chsh-def}
EXP = |<<A_1,B_1>> + <<A_1,B_2>> + <<A_2,B_1>> - <<A_2,B_2>>|,
\end{equation*}
\noindent
where $<<x,y>>$ denotes the expectation value for a collection of joint measurements of $x$ and $y$.  If $EXP > 2$, classical data reported by $A$ and $B$ violate the CHSH inequality, indicating entanglement between $q_A$ and $q_B$ \cite{chsh:69}.  For single qubits, the upper limit is $EXP \leq 2 \surd 2$, the relevant Tsirelson bound \cite{cirelson:80}.

By assuming that $C$ has the computational resources to obtain the relevant classical data from $A$ and $B$ and compute Eq. \eqref{chsh-def}, we have assumed that the interaction $H_{CE}$ is large enough to provide $C$ with the required thermodynamic free energy \cite{parrondo:15}.  We also assume that $C$ can turn on, or off, a ``decohering'' component $H_{dec}$ of $H_{CE}$ such that when $H_{dec} \neq 0$, classical data obtained from $A$ and $B$ satisfy $EXP \leq 2$, but when $H_{dec} = 0$, classical data obtained from $A$ and $B$ are such that $2 < EXP \leq 2 \surd 2$.  

Recall from above that a MIP* machine requires independent provers that communicate classically with $C$, i.e. Diagram \eqref{locc-C-diag} represents the interaction of $C$ with a MIP* machine only if $A$ and $B$ are separable, i.e. $|AB \rangle = |A \rangle |B \rangle$ for all occupied states $|A \rangle, ~|B \rangle$.  We can therefore ask whether $C$ can decide operationally, i.e. based on data received from $A$ and $B$, whether this condition is met.  

We first note an important ambiguity in the classical data received by $C$.  Let $E_C$ be the total environment with which $C$ interacts, i.e. the composite system $E_C = EAB$.  From Diagram \eqref{locc-C-diag}, we clearly have $H_{CE} = H_{CE_C}$.

\begin{lemma} \label{ambiguity}
$C$ cannot distinguish data $A_i$, $B_i$ sent by $A$ and $B$ via a classical channel from measurements of $E_C$ using observables $\hat{A}_i$, $\hat{B}_i$ that yield outcomes $A_i$, $B_i$.
\end{lemma}

\begin{proof}
Let $c_A$ and $c_B$ be the degrees of freedom of the classical channel with which $C$ directly interacts using $\hat{A}_i$ and $\hat{B}_i$, respectively.  The classical channel is a component of $E$, so $c_A$ and $c_B$ are degrees of freedom of $E$ and hence degrees of freedom of $E_C$.  $C$ can determine by measurement whether violations of the CHSH inequality by the data $A_i$, $B_i$ correlate with turning on, or off, the decohering interaction $H_{dec}$ with $E$, but $C$ cannot determine the internal interaction $H_E$ or measure the entanglement entropy $\mathcal{S}(E_1, E_2) = - \mathrm{Tr}[\mathrm{Tr}_{E_2}(\rho_{E_1, E_2}) \mathrm{ln}(\mathrm{Tr}_{E_2}(\rho_{E_1, E_2}))] $ across any decompositional boundary separating components $E_1$ and $E_2$ entirely within $E$.  Hence $C$ cannot demonstrate by measurement that the degrees of freedom $c_A$ and $c_B$ are coupled to any components $A$, $B$ of $E_C$ that do not include $c_A$ or $c_B$.
\end{proof}

The fact that all instances of classical communication require a quantum measurement, by the receiving system, of some physical encoding of the communicated information has previously been emphasized by Tipler \cite{tipler:14} among others.  Hence, we have, using the reasoning employed for Theorem \ref{main-th}:




\begin{theorem} \label{channel-th}
An observer $C$ embedded in an environment $E$ cannot determine, either by monitoring classical communication between $A$ and $B$, or by performing local measurements within $E$, whether or not $A$ and $B$ are employing a LOCC protocol with classical and quantum channels traversing $E$.
\end{theorem}

Theorem \ref{channel-th} follows immediately from Lemma \ref{ambiguity} above:

\begin{proof}
The construction of Diagram \eqref{locc-C-diag} provides $C$ with three items of data: the value of $H_{dec}$ that $C$ sets and the classical data $A_i$ and $B_i$ obtained by measuring $c_A$ and $c_B$, from which a value of $EXP$ can be computed.  We assume that $C$ computes $EXP$ using these data. There are two relevant cases: either $EXP\leq2$ or $EXP>2$, the latter of which is realized if $H_{dec} = 0$.  If $EXP \leq 2$, $C$ can infer that $A$ and $B$ are either classically correlated, which does not require a LOCC protocol since it does not require an operational quantum channel, or $A$ and $B$ are correlated through entangled pairs that respect the CHSH inequality, e.g. Werner states in appropriate parameter ranges \cite{werner}. Hence in this case, $C$ cannot determine whether $A$ and $B$ are not employing a LOCC protocol, or employing a LOCC protocol without violating the CHSH inequality. If, on the other hand, $EXP > 2$, $C$ can infer unambiguously that $A$ and $B$ share a quantum channel. However, $C$ cannot determine that $A$ and $B$ meet the separability condition $|AB \rangle = |A \rangle |B \rangle$ required by LOCC, as $EXP > 2$ is compatible with $A$ and $B$ being entangled (i.e. $|AB \rangle \neq |A \rangle |B \rangle$), which violates the conditions for a LOCC protocol.  Hence $C$ cannot determine whether $A$ and $B$ are employing a LOCC protocol, regardless of the value of $EXP$ that $C$ computes from the available data.
\end{proof}

The fact that $C$ cannot distinguish, on the basis of reported observational outcomes, between $A$ and $B$ sharing an entangled state and $A$ and $B$ being components of an entangled state is indeed well known, and is often discussed in terms of ``conspiracy'' or superdeterminism \cite{hofer:14}.  Treating $A$ and $B$ as ``effectively classical'' experimenters jointly manipulating an entangled state while remaining separable from each other -- as required by the definition of LOCC -- amounts, therefore, to a ``for all practical purposes (FAPP)'' \cite{bell:90} assumption, not a demonstrable fact; see \cite{grinbaum:17} for a general discussion.  

To see that there is no dependence of the above on the definition of $C$, consider the limit in which $C \rightarrow E$, in which case $C$ has maximal direct access to $A$ and $B$, and recall the general notion of entanglement entropy for any system $X$:

\begin{equation}\label{entanglement-1}
\mathcal{S}(X) =_{def} \mathrm{max}_{X_1, X_2 | X_1 X_2 = X} ~\mathcal{S} (X_1 X_2)\,,
\end{equation}
\noindent
where $\mathcal{S} (X_1 X_2)$ is defined as in Lemma~\ref{ambiguity}. 
In Diagram \eqref{locc-C-diag}, $E$ has no means of determining the location of the boundary between $A$ and $B$.  Hence we have:

\begin{lemma} \label{entanglement-entropy}
$E$ cannot determine the entanglement entropy $\mathcal{S}(AB)$.
\end{lemma}

\begin{proof}
For details, see the discussion and proof of \cite[Thm. 3.1]{fg:23}. Briefly, separability of $E$ from the joint system $AB$ requires a weak interaction between the two, and specifically that $N = \mathrm{ln~dim}(\mathcal{H}_{\mathscr{B}}) \ll \mathrm{ln~dim}(\mathcal{H}_E) , \mathrm{ln~dim}(\mathcal{H}_{AB})$.  Therefore $\mathscr{B}$ cannot encode, and hence $E$ cannot measure, $\mathrm{dim}(\mathcal{H}_{AB})$.  Therefore $E$ cannot measure the entanglement entropy of any decomposition of the joint system $AB$.
\end{proof}

Hence we have an alternative proof of Theorem \ref{channel-th}:

\begin{proof}
{\rm (Theorem \ref{channel-th})}  Consider the boundary $\mathscr{B}_C$ separating $C$ from the rest of $E$, and let $W$ denote everything outside of $\mathscr{B}_C$.  Then by Lemma \ref{entanglement-entropy}, $C$ cannot determine $\mathcal{S}(W_i W_j)$ between any components $W_i$ and $W_j$ of $W$.  Hence $C$ cannot detect any quantum channel in $W$, whether between $A$ and $B$ or between any other pair of subsystems of $W$.
\end{proof}

No component $C \subseteq E$, therefore, can determine $\mathcal{S}(AB)$.  Hence $C$ cannot determine, either by monitoring classical communication between $A$ and $B$ or by performing local measurements, that $A$ and $B$ are separable, i.e. $C$ cannot operationally distinguish between a MIP* machine and a monolithic quantum computer.  Any claim that a MIP* machine has solved a TM-undecidable problem is, therefore, circular, as the problem of deciding whether a physical system is a MIP* machine is itself TM-undecidable.


\section{Closed Timelike Curves}\label{ctcs}

We now look at a similar situation of non-operational identifiability in a setting in which {\em closed timelike curves} (CTCs) are allowed as computational resources. The idea of CTCs evolved from a number of cosmological questions, particularly pertaining to Black Hole theory, such as those concerning the construction and stability of ER-bridges \cite{morris:88,hawking:92} (for a historical survey, see e.g. \cite{luminet:21}). When instrumental in models of classical computation, CTCs make it possible to solve hard computational problems in constant time (surveyed in \cite{brun:03}). David Deutsch \cite{deutsch:91} demonstrated that quantum computation with quantum data which is capable of traversing CTCs provided a new and powerful physical model of computation, along with self-consistent evolution further engendering (quantum) computational complexity \cite{bacon:04}. As pointed out in \cite{aaronson:09}, Deutsch's approach was to treat a CTC as a region of spacetime where a `causal consistency' condition is imposed; specifically, a region in which the time-evolution operator maps state of the initial hypersurface to itself. This initial state is, therefore, a probabilistic fixed point of the time-evolution operator within the CTC, i.e. a state $\rho$ such that $\Phi (\rho) = \rho$ for the time-evolution operator $\Phi$ within the CTC. 

To model computation using a CTC, consider a Hilbert space of qubits given by $\mathcal{H} = \mathcal{H}_{\rm{ch}} \otimes \mathcal{H}_{\rm{tv}}$, where $\mathcal{H}_{\rm{ch}}$ denotes that of the chronologically respecting qubits, and $\mathcal{H}_{\rm{tv}}$ that of those which traverse CTCs, as shown in Diagram \eqref{ctc-diag} ($cf.$ \cite[Fig. 3]{deutsch:91}):

\begin{equation} \label{ctc-diag}
\begin{gathered}
\begin{tikzpicture}
\node at (4.1,.6) {Timelike data path ($\mathcal{H}_{ch}$ qubits)};
\node at (-4.8,0.3) {CTC data path};
\node at (-4.7,-0.3) {($\mathcal{H}_{tv}$ qubits)};
\draw [thick, ->] (-1,0.2) arc [radius=1, start angle=9, end angle= 351];
\draw [thick, ->] (1,-1) -- (1,1);
\filldraw [fill=white, draw=black] (-2,-0.1) rectangle ++(4,0.4);
\node at (-0.1,0.1) {Process};
\end{tikzpicture}
\end{gathered}
\end{equation}

Importantly, the evolution of the CTC qubits is determined by self-consistency --- though the qubits themselves are an expendable resource \cite{bacon:04}. This means that the state of the CTC qubits at the temporal origin should be the same as those qubits after the evolution $\mathbf{U}$ operator corresponding to the Process in \eqref{ctc-diag}. The density matrix $\mathbf{\rho}$ as a solution at the former, is given by
\begin{equation}\label{density-matrix-1}
\mathbf{\rho} = {\rm{Tr}}_{\rm{ch}} [\mathbf{U} (\mathbf{\rho}_{\rm{in}} \otimes \mathbf{\rho} )\mathbf{U}^{\dagger} ]\,,
\end{equation}
where $\mathbf{\rho}_{\rm{in}}$ denotes the density matrix of the chronologically respecting qubits, and $\rm{Tr}_{\rm{ch}}$ denotes the trace of $\mathcal{H}_{\rm{ch}}$ \cite{deutsch:91,bacon:04}. Thinking in terms of a quantum circuit, and the solution in \eqref{density-matrix-1}, the output $\mathbf{\rho}_{\rm{out}}$ of the circuit is given by \cite{deutsch:91,bacon:04}:
\begin{equation}\label{density-matrix-2}
\mathbf{\rho}_{\rm{out}} = {\rm{Tr}}_{\rm{tv}}[\mathbf{U} (\mathbf{\rho}_{\rm{in}} \otimes \mathbf{\rho} )\mathbf{U}^{\dagger} ]\,.
\end{equation}
This supposes a `gating-free' system. If gating is applied, then the consistency condition changes, and a previously selected temporal origin now becomes arbitrary. It is shown, however, in \cite{bacon:04} that 
potentially different self-consistency solutions are relatable via a standard change of basis.

Let us now reconsider Diagram \eqref{locc-C-diag}, treating the joint system $AB$ as an arbitrary quantum computer and setting the decohering interaction $H_{CE}$ to zero, or equivalently, using the result of Theorem \ref{main-th} to treat $Q$ as an internal quantum resource used by $AB$.  We can then ask: what can $C$ infer about the computational role of the classical channel connecting the ``components'' $A$ and $B$?  This channel being classical requires, by definition, that it is timelike as measured by clocks in $E$. Taking the limit as $C \rightarrow E$, the channel is timelike as measured by $C$'s clocks.  Classicality for $C$ also requires that the channel has finite length, i.e. the endpoints of the channel, which we can denote $A_c$ and $B_c$ respectively, must be such that ${\rm d}_C (A_c, B_c) > 0$ in $C$'s distance metric ${\rm d}_C$.  However, from Lemma \ref{entanglement-entropy} above, we have that $C$ cannot determine the entanglement entropy of any state $|AB \rangle$. Hence $C$ cannot determine that $A$ and $B$ are separable as discussed above.  In particular:

\begin{lemma} \label{no-metric}
In any physical setting described by Diagram \eqref{locc-C-diag}, $C$ cannot determine the distance ${\rm d}_{AB} (A_c, B_c)$, where ${\rm d}_{AB}$ is the metric employed by $AB$, between the classical channel endpoints $A_c$ and $B_c$ on $\mathscr{B}$.
\end{lemma}

\begin{proof}
The systems $E$ and $AB$ are mutually separable in Diagram \eqref{locc-C-diag} by construction, so the result follows from the requirement that mutually separable systems have independent, free choice of QRFs, including space and time QRFs \cite{fgm:22b}.
\end{proof}

From Lemma \ref{no-metric}, we immediately have:

\begin{theorem} \label{ctc}
In any physical setting described by Diagram \eqref{locc-C-diag}, $C$ cannot determine whether $AB$ employs CTCs as computational resources.
\end{theorem}

\begin{proof}
From Lemma \ref{ctc}, $S$ cannot show that ${\rm d}_{AB} (A_c, B_c) \neq 0$.  If ${\rm d}_{AB} (A_c, B_c) = 0$, however, the classical channel from $A_c$ to $B_c$ in $E$ is a CTC for $AB$, and hence is available to $AB$ as a computational resource.
\end{proof}

Aaronson and Watrous \cite{aaronson:09} have shown that both classical and quantum computers can employ CTCs to solve any problems in the complexity class PSPACE --- this consists of all problems solvable on a classical TM with a polynomial amount of memory. Theorem \ref{ctc} shows that the problem of deciding whether a physical system is a computer that can employ CTCs as a resource is TM-undecidable. Thus 
a proffered solution to a PSPACE problem, for which independent means of verification are unavailable, is TM-undecidable.

\section{Discussion} \label{concl}

We have shown here that whether quantum, or in the case of CTCs even classical, computers employ nonlocal resources when performing computations is generically undecidable in operational settings.  All interactions with physically-implemented computers are operational; hence our results apply to all such interactions.  They show that the space and time complexity of physically-implemented computational processes cannot be determined unambiguously, and place principled limits on the extent to which formal descriptions of computational processes, e.g. formal descriptions of MIP* or CTC-using machines, can be demonstrably realized in practice.  They also limit our ability to infer from observations and experiments the computational architectures of computers found ``in the wild'', including living organisms.


As shown in \cite{culf:24,mastel:24}, constraint systems (CS) and CSPs can be formulated in the language of MIP and MIP* architectures, with the verifier $C$ implementing the satisfaction condition. Specifically, \cite[\S 4]{culf:24} and \cite[Th. 1.1]{mastel:24} demonstrate relations between CSPs, languages in MIP* (and hence in RE), and protocols for the Halting problem of the form CS-MIP*(2,1,$c$,$s$), with $c$ and $s$ being the completeness and soundness probabilities, respectively; see \cite[Cor. 4]{culf:24} for the special case where $c=1$.  The results of \S \ref{main-result} show that $C$ cannot operationally demonstrate independence between constraints and identified partial solutions; this applies to protocols of the form  CS-MIP*(2,1,$c$, $s$) as special cases. In fact, the Halting problem has been shown \cite{df:20} to be equivalent to the Frame problem \cite{mccarthy:69}: broadly speaking, the problem of circumscribing whatever is relevant in a given physical situation. What we have shown here is, in essence, that empirically circumscribing resource availability and usage requires solving the Frame problem.

These results can be given a straightforward interpretation: finite interactions with an unknown quantum system can place a lower limit, but not an upper limit, on the Hilbert-space dimension of that system.  This extends to quantum systems the limitations on inferences from finite observations proved for classical systems in 1956 \cite{moore:56}.  The existence of such limits illustrates the profound distinction between behaviors that can be shown theoretically to be logically possible and behaviors that can be unambiguously observed by finite agents such as ourselves.


\end{document}